\documentclass[11pt]{article}
\usepackage{mystyle}
\usepackage{natbib}

\usepackage{algorithm}
\usepackage{algpseudocode}

\newtheorem{Alg}{Algorithm}
\newcommand{\myalg}[4][0cm]{
\medskip
\small{
\fbox{
\parbox{5.2in}{\vspace{#1}
\begin{Alg}\label{#2}{\textsl{ #3}}
\vspace{.1cm}\\ \emph{ #4}
\end{Alg}
}}
\medskip
}}

\DeclareMathOperator{\cost}{cost}
\newcommand{\ctopk}{\cost_{k}}
\newcommand{\Bin}{B_{\text{in}}}
\newcommand{\Gin}{G_{\text{in}}}
\newcommand{\Bout}{B_{\text{out}}}
\newcommand{\bout}{b_{\text{out}}}
\renewcommand{\bot}{\operatorname{bottom}}
\renewcommand{\top}{\operatorname{top}}
\DeclareMathOperator{\distortion}{distortion}

\DeclareMathOperator{\robustness}{robustness}
\DeclareMathOperator{\consistency}{consistency}

\newcommand{\Mopt}{M^*}
\newcommand{\itemcolor}{black}
\newcommand{\agentcolor}{cyan}

\begin{document}
\title{Optimal Metric Distortion for Learning-Augmented Matching on the Line}
\author{
Jabari Hastings\thanks{Supported by the Simons Foundation Collaboration on the Theory of Algorithmic Fairness and the Simons Foundation Investigators Award 17351.}\\
Stanford University
\and
Marena Richter\thanks{Funded by the Deutsche Forschungsgemeinschaft (DFG, German Research Foundation) – 459420781 (FOR AlgoForGe) and affiliated with the Lamarr Institute for Machine Learning and Artificial Intelligence. }\\
University of Bonn
}

\date{\vspace{15pt}\small{\today}}

\maketitle

\begin{abstract} 
We revisit the problem of matching on the line with ordinal preferences.
In the classic setting, there are $n$ agents and $n$ items in a shared unknown line metric,
and the goal is to find a low-cost perfect matching using only the agents' rankings of 
the items by distance. 
A mechanism has \emph{distortion} $\alpha$ if it always outputs a matching whose cost 
is within a factor of $\alpha$ of the optimum, in every consistent line metric. 

In the \textit{learning-augmented} setting, the mechanism is also supplied with a prediction 
that conveys additional information about the instance. 
The quality of this prediction is unknown, and the goal is to optimize the mechanism's
distortion when the prediction is accurate
(\emph{consistency}), while preserving worst-case guarantees when the prediction is arbitrarily
inaccurate (\emph{robustness}). We propose a mechanism that takes a matching as its prediction
and guarantees $1$-consistency and $3$-robustness. By recovering an optimal
matching when the prediction is perfectly accurate while retaining the optimal prediction-free
distortion guarantee when it is arbitrarily inaccurate, we resolve an open question of
Filos-Ratsikas et al. (IJCAI, 2025).
\end{abstract}

\thispagestyle{empty}
\clearpage 

\section{Introduction}
Matching is a fundamental assignment task that underpins a host of high-stakes, real-world scenarios. To see its profound societal impact, one needs to look no further than the school choice problem. In many countries, the placement of children into schools relies on centralized matching mechanisms. Given the outsize role of schooling in a child's development, these assignments are frequently contentious. Describing the pressure these systems place on families, \citet{Tobin:2003aa}
wrote of one U.S. school district:
\begin{quote}
\textit{``This complicated new system requires you to be a savvy parent. Some parents have treated the choice application with a casualness befitting a form for summer camp. Yet it requires the same level of care you would give an IRS tax return.''} \\
--- Thomas Tobin, \textit{St. Petersburg Times}, 2003
\end{quote}

Part of this anxiety stems from the informational constraints of the matching mechanisms often deployed, such as the \emph{Boston Mechanism} \citep{Abdulkadiroglu:2003aa},
\emph{Deferred Acceptance} \citep{Gale:1962aa}, and
\emph{Random Serial Dictatorship} \citep{Abdulkadiroglu:1998aa}. These algorithms typically receive only ordinal information (e.g., rankings), while the cardinal social costs they are meant to optimize are often latent or under-specified. Given this limitation, it seems reasonable to expect them to be less efficient than an algorithm with complete information.

In computational social choice, an ordinal mechanism's loss in efficiency is often analyzed under the \textit{distortion} framework, introduced by \citet{Procaccia:2006aa} and surveyed by \citet{Anshelevich:2021aa}. Here, an ordinal mechanism is evaluated in terms of its worst-case approximation ratio of the optimal social cost (called the \textit{distortion}), over all cardinal values  consistent with a given set of rankings. Since the approximation ratio could be unbounded when the underlying cardinal values are unrestricted, the literature often adds mild constraints to make the analysis tractable and meaningful. A particularly fruitful direction is the \textit{metric distortion} model \citep{Anshelevich:2015aa}, where ordinal preferences are assumed to arise from distances in an unknown metric space. The goal is then to design mechanisms that, using only these rankings, achieve low social cost for every metric consistent with the observed preferences.
Regardless of whether the distances represent literal commutes to a school or simply abstract proximities in preference, this simultaneous approximation serves as a powerful robustness guarantee.

Applying the metric distortion framework to matching problems has revealed how challenging it is to optimize without cardinal information. In their seminal work, \citet{Caragiannis:2016aa,Caragiannis:2024ab} showed that for the objective of minimizing total cost, the \textit{Serial Dictatorship} has a worst-case distortion of $O(2^n)$, though its variant \textit{Random Serial Dictatorship} has expected distortion $O(n)$. Subsequent work has not yet improved beyond this linear guarantee. \citet{Anari:2023aa} established a lower bound of $\Omega(\log n)$ and gave the deterministic \textit{RepMatch} mechanism with distortion $O(n^2)$, later adapted by \citet{Hastings:2025aa} to broader fairness objectives. Even the recent derandomization of \textit{Random Serial Dictatorship} by \citet{Haqi:2026aa} remains at the linear threshold, with distortion $\widetilde{O}(n)$.  Achieving sublinear distortion remains a notoriously challenging open question in the field.

One pathway to achieving stronger distortion guarantees is to understand more deeply how the expressiveness of the underlying metric affects the performance of ordinal mechanisms. A natural starting point for this inquiry is the line metric, a classic abstraction for economic and political alignment \citep{Hotelling:1929aa, Downs:1957aa}. In this domain, the literature has seen sharp progress: \citet{Filos-Ratsikas:2025aa} demonstrated that a specialized deterministic algorithm can achieve a constant distortion of exactly 3, which they also proved to be optimal. In contrast, all known mechanisms for the general metric problem incur distortion $\Omega(n)$ on some line instance.

Another strategy to bypass worst-case barriers is to supplement a mechanism with a modest amount of additional information about the problem instance. One increasingly prominent approach to formalizing this idea is the \textit{learning-augmented} framework, which explores how predictions can improve algorithmic performance \citep{Mitzenmacher:2022aa}. Such advice is highly practical: in real-world scenarios, modern machine learning models or domain expertise can often provide useful, albeit partial, estimates of latent parameters.
In a school choice setting, for example, historical application patterns and geographic distances could provide a reasonable estimate of a family's underlying costs.
Motivated by this potential, predictions have recently been integrated into classic mechanism design settings \citep{Agrawal:2022aa,Xu:2022} and, notably, into the voting problem \citep{Berger:2024aa}. Across all these domains, the fundamental challenge is designing mechanisms that satisfy two goals in tension: \textit{consistency} (achieving near-optimal efficiency when the predictions are highly accurate) and \textit{robustness} (maintaining a strict worst-case baseline guarantee even if the advice is entirely erroneous).

We explore the intersection of these two promising directions by investigating matching with predictions on the line. Here, the goal is to design a mechanism that achieves near-perfect consistency while having distortion no worse than prediction-free mechanisms---a direction recently suggested by \citet{Filos-Ratsikas:2025aa}. However, the extent to which predictions can reduce distortion below 3 remains unclear. Recent work on voting by \citet{Berger:2024aa} provides a cautionary precedent: they proved that no deterministic voting rule can simultaneously be $\frac{3 - \delta}{1 + \delta}$-consistent and better than $\frac{3 + \delta}{1 - \delta}$-robust for any $\delta \in [0, 1)$, even when the metric is restricted to line instances. For example, a $2$-approximation for consistency would require a $5$-approximation for robustness, far worse than the distortion of 3 achievable without predictions \citep{Gkatzelis:2020aa, Kizilkaya:2022aa, Kizilkaya:2023ab}. In light of this, we ask: do matching mechanisms fare any better?

\subsection{Our Contributions}
We consider deterministic mechanisms that are augmented by a predicted matching $\widehat{M}$. This design choice is in line with learning-augmented mechanism design, where predictions are often made in the outcome space \citep{Agrawal:2022aa}. This is also natural in practice: such predictions could arise from historical assignment data, statistical models of the latent metric, or, more plainly, a guess of the optimal outcome.  

Our main contribution is designing a mechanism that incorporates the predicted
matching to achieve 1-consistency and 3-robustness
(\Cref{thm:prediction-best-of-both}), resolving an open question of
\citet{Filos-Ratsikas:2025aa}. This guarantee holds for every top-$k$ cost,
and therefore, by majorization, extends to the fairness ratio: distortion with
respect to any monotone symmetric norm \citep{Goel:2017aa, Goel:2018aa}.
We also show a refined distortion guarantee: if the prediction $\widehat{M}$
has distortion $\eta$, then our augmented mechanism has distortion
$\min\{\eta,3\}$.

En route to proving this best of both worlds guarantee, we introduce a prediction-free mechanism that achieves optimal distortion on the line (\Cref{thm:reserved-outliers}). Our mechanism is inspired by but is more general than the \textit{OrderMatch} mechanism \citet{Filos-Ratsikas:2025aa}, which is optimal on the line. We show that a significantly larger family of matchings has low distortion, giving us a convenient bounding box for selecting based on the prediction.

\subsection{Related Work}
\paragraph{Learning-Augmented Framework.} In recent years, algorithms with predictions have garnered substantial attention. Since the formulation of the learning-augmented framework by \citet{Lykouris:2021}, there have been numerous applications to problems in social choice and beyond, such as voting \citep{Berger:2024aa,Filos-Ratsikas:2025ab}, secretary problems \citep{Dutting:2021,Antonios:2023}, and facility location \citep{Agrawal:2022aa,Xu:2022}. We refer the reader to \citep{Mitzenmacher:2022aa} for a survey of earlier work on algorithms with predictions, and to \citep{Lindermayr} for a catalog of recent work.

\paragraph{Ordinal Matching.} Since the early work of \citet{Hylland:1979aa}, a long line of work has studied the problem of matching with ordinal information. \citet{Filos-Ratsikas:2014aa,Amanatidis:2022aa} analyze the utilitarian distortion of mechanisms, providing optimal distortion bounds of $O(\sqrt{n})$ for randomized mechanisms (under unit-sum and unit-range valuations) and $O(n^2)$ for deterministic mechanisms (under unit-sum valuations). This utilitarian setting is a maximization problem, in contrast to the metric minimization setting studied here.

A related direction studies how much cardinal information is needed to improve distortion. \citet{Amanatidis:2022aa,Amanatidis:2024aa} show that querying a small number of values per agent yields polynomial improvements, including $O(\sqrt{n})$ distortion with two queries per agent and $O(n^{1/(\lambda+1)})$ distortion with $O(\lambda\log n)$ queries per agent. This trade-off was later sharpened by \citet{Ebadian:2025aa}, who obtain the optimal $O(n^{1/\lambda})$ distortion using $\lambda$ queries per agent. Other work considers range queries on agent utilities \citep{Ma:2021aa,Latifian:2025aa}. Recently, \citet{Gkatzelis:2026aa} showed that metric distortion can fall below $3$ when ordinal preferences are combined with approval information and inter-item distances.

Metric distortion has also been applied to maximum matchings with ordinal information under metric assumptions. In contrast to our minimization setting, the distortion bounds there are constant for both truthful and randomized mechanisms \citep{Anshelevich:2016aa,Anshelevich:2016ab}.

\section{Preliminaries}\label{sec:prelims}

We let $A = \{a_1, a_2, \dots, a_n\}$ denote the set of agents, and let $B = \allowbreak \{b_1, b_2, \dots, b_n\}$ denote the set of items, where $n\geq 2$ is some integer. 
A \emph{problem instance}, denoted by $\sigma$, is given by $n$ \emph{preference lists}, one provided by each agent. Agents in $A$ provide a permutation of the items in $B$; permutations are given in non-decreasing order of distance. 

For items $b_i,b_j\in B$, we write $b_i \succ_a b_j$ if agent $a$ ranks $b_i$ above $b_j$,
and write $b_i \succeq_a b_j$ if $b_i=b_j$ or $b_i \succ_a b_j$; the relations
$\prec_a$ and $\preceq_a$ are defined analogously.
Given a set $S \subseteq B$, we let $\top_a(S)$ and $\bot_a(S)$ denote agent $a$'s most preferred and least preferred item, respectively, in the set $S$. 

\paragraph{Metric spaces.} 
Agents and items are all assumed to lie on the real line, $\mathbb R$. 
Their distances are conveyed by an unknown \emph{metric} \allowbreak $d:(A\cup B) \times (A\cup B) \to \mathbb{R}_{\geq 0}$. 
By definition, the function $d$ satisfies the following properties.

\begin{enumerate}[label=(\roman*),noitemsep]
	\item \textit{Identity of indiscernibles:} $d(x, y) = 0$ iff $x = y$,\footnote{As in \citep{Filos-Ratsikas:2025aa}, we do not allow collocation of agents or items in the line metric.}

	\item \textit{Symmetry:} $d(x, y) = d(y, x)$,

	\item \textit{Triangle inequality:} $d(x, y) \leq d(x, z) + d(z, y)$.
\end{enumerate}
We say that a metric $d$ is \emph{consistent} with an instance $\sigma$ if, for every agent $a\in A$ and items $b,c\in B$, $b\succ_a c$ implies $d(a,b)\leq d(a,c)$. We let $\rho(\sigma)$ denote the set of all line metrics consistent with the problem instance. 
Throughout, we restrict attention to instances $\sigma$ with
$\rho(\sigma)\neq\emptyset$.

Since agents and items are on a line, we slightly abuse notation by letting agent and item names also denote their coordinate positions on $\mathbb{R}$. 
For instance, $b < b_i$ indicates that item $b$ is strictly to the left of item $b_i$, and $a \leq b_r$ indicates that agent $a$ lies no farther right than item $b_r$.

\paragraph{Matchings.} A matching $M$ is a set of disjoint pairs
$\{a,b\}$ with $a\in A$ and $b\in B$. When $\{a,b\}\in M$, we write
$M(a)=b$ and $M(b)=a$. Thus, we view a matching as an involution on the
matched agents and items whenever convenient. When every agent and every item
appears in exactly one pair, we say that $M$ is a \emph{perfect matching}.

\paragraph{Distortion.}
We define the top-$k$ \textit{cost} of a matching $M$ to be the sum of the $k$ largest distances between agents and their assigned items: 
\begin{equation*}
\ctopk(M, d) \coloneqq \max_{S \subseteq A, |S| = k} \sum_{a_i \in S} d(a_i, M(a_i)). 
\end{equation*}
We define the top-$k$ \textit{distortion} of a perfect matching $M$ with respect to a specific metric $d$ as the ratio of that matching's cost compared to that of the optimal matching:
\begin{equation*}
\distortion_k(M, d) \coloneqq \frac{\ctopk(M, d)}{\min_{M'} \ctopk(M', d)}.
\end{equation*}
The worst-case top-$k$ distortion of a matching for a given instance $\sigma$ is given by:  
\begin{equation*}
\distortion_k(M, \sigma) \coloneqq \sup_{d \in \rho(\sigma)} \distortion_k(M, d).
\end{equation*}

\paragraph{Matching with predictions.}
A \emph{learning-augmented mechanism} is any algorithm $F(\sigma, \widehat{M})$ that takes an instance $\sigma$ and a prediction $\widehat{M}$ as input, and outputs a perfect matching. 
In this work, we assume that the prediction $\widehat{M}$ is a perfect matching.

We define two notions typically used to measure the effectiveness of mechanisms that incorporate predictions. Let $\rho(\sigma, \widehat{M}) \subseteq \rho(\sigma)$ denote the set of consistent metrics for which $\widehat{M}$ is an optimal perfect matching. The \textit{consistency} of a mechanism $F$ is its worst-case performance when the provided prediction is exactly correct:
\begin{equation*}
\consistency_k(F) \coloneqq \sup_{\sigma, \widehat{M}}  \sup_{d \in \rho(\sigma, \widehat{M})} \distortion_k( F(\sigma, \widehat{M}), d ).
\end{equation*}
The \textit{robustness} of a mechanism is its worst-case distortion over all instances and predictions, accounting for when the prediction may be arbitrarily wrong:
\begin{equation*}
\robustness_k(F) \coloneqq \sup_{\sigma, \widehat{M}}  \sup_{d \in \rho(\sigma)} \distortion_k( F(\sigma, \widehat{M}), d ).
\end{equation*}
Note that if a mechanism $F$ does not use the prediction $\widehat{M}$ at all, its consistency and robustness simply equal the worst-case distortion of the matching it produces.

\subsection{Preliminary Lemmas}

\paragraph{Optimal matching.}
In any given instance, there may be multiple matchings that optimize a particular top-$k$ objective.
In much of our analysis, we rely on the useful fact that greedily matching the leftmost agent with the leftmost item results in an optimal matching. 
\begin{lemma}[\citet{Filos-Ratsikas:2025aa}]\label{lem:assortative}
Fix a line metric $d$ and let $\pi_A: A \to [n]$ and $\pi_B: B \to [n]$ be functions mapping agents and items to their positional ranks on the line from left to right. The matching
$M = \{\{a_i, b_j\} : \pi_A(a_i)=\pi_B(b_j)\}$
is an optimal matching that minimizes the top-$k$ cost for every $k\in[n]$.
\end{lemma}
For a fixed consistent line metric $d$, we let $\Mopt$ denote the assortative
matching from \Cref{lem:assortative}, which is guaranteed to be optimal for every
top-$k$ objective.

\paragraph{Crossings and swaps.}
For a matching $M$ and agents $a_i,a_j$ with $a_i<a_j$, we say their
assignments are \textit{crossed} if $M(a_j)<M(a_i)$, where inequalities refer
to the true line order. The \textit{$(a_i,a_j)$-swap} of $M$ is the matching
obtained by replacing
$\{a_i,M(a_i)\},\{a_j,M(a_j)\}$ with
$\{a_i,M(a_j)\},\{a_j,M(a_i)\}$; if the assignments are crossed, we call this
swap an \textit{uncrossing}. The next lemma, proved in \Cref{app:prelims}, shows
that uncrossing can only improve the top-$k$ cost.

\begin{restatable}[Uncrossing weakly decreases cost]{lemma}{uncrossinglemma}
\label{lem:uncrossing}
Let $M'$ be the result of a $(a_i,a_j)$-swap on a matching $M$ with
$a_i<a_j$ and $M(a_j)<M(a_i)$. Then
$\ctopk(M',d)\leq \ctopk(M,d)$ for every $k\in [n]$.
\end{restatable}

\section{Matching without Predictions}\label{sec:line}
As a precursor to our results for matching with predictions, we revisit the traditional (prediction-free) setting of matching on the line. 
Here, it is known that the \textit{OrderMatch} mechanism of \citet{Filos-Ratsikas:2025aa} attains distortion at most 3 with respect to any top-$k$ cost objective, which is optimal among ordinal matching mechanisms. The high-level approach of \textit{OrderMatch} is quite natural. Motivated by  \Cref{lem:assortative}, the mechanism aims to recover an ordering of agents and items and then greedily matches them in an assortative fashion. There may be instances where the mechanism fails to determine the relative location of certain outlier items. Nevertheless, by first greedily matching ordered items with ordered agents, and matching the outliers last, the resulting matching cannot have large distortion.

In this section, we strengthen that result to show that an even larger family of assortative matchings may have distortion at most 3. In fact, any matching obtained by first \textit{arbitrarily} assigning the
outliers and then assortatively matching the remaining items will also have distortion at most 3. This freedom will be important in \Cref{sec:predictions}, where the predicted
matching will determine which agents receive the outliers without sacrificing
the worst-case distortion guarantee.

\subsection{OrderMatch with Reserved Outliers}\label{sec:order-match}
Our approach builds on that of \citet{Filos-Ratsikas:2025aa} and follows a similar sequence of steps. We first identify a pair of extreme items, which serve as anchors for the agent/item ordering we construct. We then recover an ordering for as many items as can be gleaned from the agents' preferences. Finally, we use this ordering to match agents to items.

At a finer level of detail, however, our approach differs in several subtle ways. To present the strengthened guarantees of \textit{OrderMatch} cleanly, we define certain components differently. We highlight these differences as they arise.

\paragraph{Finding extreme items.}
Consider the set of items that are most preferred by at least one agent,
\begin{equation*}
    B_+ \coloneqq \{ \top_a(B) : a \in A \}.
\end{equation*}
Observe that the set of items ranked last by agents when restricted to $B_+$, defined as 
\begin{equation*}
    B_{\text{ext}} \coloneqq \{ \bot_a(B_+) : a \in A \},
\end{equation*}
has cardinality at most two: the only possible elements are the leftmost and
rightmost items in $B_+$. Indeed, any other item $b\in B_+$ lies strictly
between these two endpoint items and thus cannot be ranked last by any agent.
We denote the leftmost and rightmost items in $B_+$ under the true line order by
$b_\ell$ and $b_r$, respectively. This labeling is only for the analysis; the
algorithms presented in this work need only identify the unordered set of
extremes.

\paragraph{Ordering items on the line.}
The items $b_\ell$ and $b_r$ serve as convenient anchors, enabling us to recover an ordering of some of the items on the line. 

\citet{Filos-Ratsikas:2025aa} show that when agents $a_\ell$ and $a_r$  with $\top_{a_\ell}(B) = b_\ell$ and $\top_{a_r}(B)  = b_r$ are carefully chosen, one can recover the ordering of many of the items on the line that are located between $b_\ell$ and $b_r$ or are otherwise close to one of these extreme items.

Motivated by this insight, we define the set $\Bin$\footnote{Our definition differs from that of
\citet{Filos-Ratsikas:2025aa}, who define
$\Gin(a_i,a_j)=\{b\in B: b\succ_{a_i} b_\ell \text{ or } b\succ_{a_j} b_r\}$
for agents $a_i,a_j$ with $\top_{a_i}(B)=b_\ell$ and
$\top_{a_j}(B)=b_r$, chosen to maximize $|\Gin(a_i,a_j)|$. We use a predicate
over all agents rather than fixing these two witnesses, making $B\setminus\Bin$ 
the set of items every agent ranks below both extremes.}:
\begin{equation}
\Bin = \{ b \in B : b \succeq_{a} b_\ell \ \text{or} \ b \succeq_{a} b_r \ \text{for some $a \in A$} \}.
\end{equation}
We define the items outside the set $\Bin$ as the outliers:
\begin{equation}
    \Bout = B \setminus \Bin .
\end{equation}
Equivalently, every outlier is ranked below both extremes by every agent.
Intuitively, $\Bin$ is the set of items that can be sorted. If an item $b$ is not in the set $\Bin$, then every agent prefers it less than both $b_\ell$ and $b_r$. This item could in principle be far to the left or right, but one would not be able to tell from the preferences. 
In \Cref{lem:weak-item-order}, we formally justify that $\Bin$ can be ordered.
The proof is deferred to \Cref{app:order-match}; our argument is similar in spirit to order-recovery procedures
of \citep{Elkind:2014aa, Babashah:2025aa}, but we give a direct proof tailored
to the set $\Bin$.

\begin{restatable}[Recovering an Item Ordering]{lemma}{weakitemorderlemma}
\label{lem:weak-item-order}
    There is a polynomial-time algorithm that on input $\sigma$, recovers the left-to-right ordering (or reversal) of the items in $\Bin$, denoted by $\pi_B$. 
\end{restatable}

We next show why the
remaining items can safely be treated as outliers.
The proof is deferred to \Cref{app:order-match}; while straightforward, it crucially relies on the predicate formulation of $\Bin$. 
\begin{restatable}[Outlier geometry]{lemma}{outliergeometrylemma}
\label{lem:outlier-geometry}
Every outlier $\bout \in \Bout$ lies outside the interval between
$b_\ell$ and $b_r$. Moreover, for every $\bout\in \Bout$, the following hold:
\begin{enumerate}[label=(\arabic*),noitemsep]
    \item If $\bout < b_\ell$, then $\bout\leq b$ for every $b\in\Bin$.
    \item If $b_r < \bout$, then $b\leq \bout$ for every $b\in\Bin$.
\end{enumerate}
\end{restatable}

\paragraph{Ordering agents on the line.} With the item ordering
in place, we can recover a coarse ordering of the agents by sorting them
according to their top-choice items.
As in \citet{Filos-Ratsikas:2025aa}, top-choice items provide coarse location
information about agents: agents with left-to-right ordered top choices must
respect the same left-to-right order.
The proof is deferred to \Cref{app:order-match}.

\begin{restatable}[Recovering a Weak Agent Ordering]{lemma}{weakagentorderlemma}
\label{prop:weak-agent-order}
Suppose the ordering $\pi_B$ recovered by \Cref{lem:weak-item-order} is
oriented consistently with the true left-to-right ordering of $\Bin$. There is
a polynomial-time algorithm that, given $\sigma$ and $\pi_B$, returns a weak
ordering $\pi_A$ of the agents with the following properties:
\begin{enumerate}[label=(\arabic*),noitemsep]
    \item The equivalence classes of $\pi_A$ are exactly the sets
    $\{a\in A:\top_a(B)=b\}$, for $b\in B_+$.
    \item If $a_i$ strictly precedes $a_j$ in $\pi_A$, then $a_i<a_j$ in the
    true left-to-right order of the agents.
\end{enumerate}
\end{restatable}

\paragraph{Ordered matching.}
Finally, we assortatively match agents and items. Our mechanism first
reserves the outlier items for arbitrary agents and then order-matches the
remaining agents to the sortable items. The formal description of our mechanism is given in \Cref{alg:reserved-outliers}.

\begin{center}
\myalg{alg:reserved-outliers}{OrderMatch with Reserved Outliers}{
\textbf{Input:} A preference table $\sigma$. 
\begin{enumerate}[noitemsep,leftmargin=*]
    \item Compute $\Bin$, $\Bout$, the order $\pi_B$ of $\Bin$, and the weak
    agent order $\pi_A$.
    \item Choose an arbitrary set $A_{\mathrm{res}} \subseteq A$ with
    $|A_{\mathrm{res}}|=|\Bout|$ and arbitrarily and bijectively match the items
    in $\Bout$ to these agents. Let this partial matching be $M_1$.
    \item Let $M_2$ be an assortative matching between
    $A\setminus A_{\mathrm{res}}$ and $\Bin$ according to $\pi_A$ and $\pi_B$.
    Ties among agents with the same top item are broken arbitrarily.
    \item Return the resulting matching $M = M_1 \cup M_2$.
\end{enumerate}
}
\end{center}
By allowing the items of $\Bout$ to be reserved for arbitrary agents, our mechanism can return a matching from a large family, including the matchings selected by \citet{Filos-Ratsikas:2025aa}. In the following theorem, we show that all of these options must have low distortion. The next section is devoted to proving this theorem. 

\begin{theorem} \label{thm:reserved-outliers}
For every instance $\sigma$ and every $k\in [n]$, any
matching $M$ returned by \Cref{alg:reserved-outliers} satisfies
$\distortion_k(M,\sigma) \leq 3$.
\end{theorem}
Prior work has already shown that the optimal distortion attainable on the line by a deterministic mechanism is $3$ \citep{Filos-Ratsikas:2025aa,Caragiannis:2024ab}. Our result illustrates that a broad class of matchings achieves this optimal distortion on the line.

\subsection{Analysis of OrderMatch}
While \Cref{alg:reserved-outliers} has a succinct description, it makes
arbitrary decisions at a few steps, resulting in a large family of matchings.
This freedom is helpful when we later incorporate predictions
(\Cref{sec:predictions}), but we now have to prove distortion guarantees for
matchings which may significantly vary in structure. To that end, we apply a
sequence of transformations to a given matching $M$ to create a canonical, i.e.,
more structured representative, $\widetilde M$ that is easier to analyze.

\paragraph{A canonical matching.}
For a consistent metric $d\in\rho(\sigma)$, we consider the set of agents whose
assignment under $\Mopt$ lies in $\Bin$:
\begin{equation}\label{eq:A_in}
    A^*_{\mathrm{in}}=\{a\in A:\Mopt(a)\in \Bin\}
\end{equation}
Recall from \Cref{prop:weak-agent-order} that the recovered order of $\Bin$
partitions $A^*_{\mathrm{in}}$ into blocks based on top-choice items. 
For our analysis, we would like the canonical matchings to be those whose
mistakes are primarily confined to agents with the same top-choice item: if
agent $a_i$ is assigned an item $M(a_i)$, then the agent optimally matched to
that item should be in the same top-choice block as $a_i$.
\begin{definition}
Given a matching $M$, we say the assignment of $M(a_i)$ to $a_i$ is \textit{internal}
if, writing $a_j=\Mopt(M(a_i))$, we have
$a_i,a_j\in A^*_{\mathrm{in}}$ and $a_i$ and $a_j$ share the same top-choice
item.
\end{definition}

This terminology is adapted from \citet{Filos-Ratsikas:2025aa}, who analyze
\textit{OrderMatch} through a permutation graph.\footnote{In that graph,
assigning $M(a_i)$ to $a_i$ corresponds to the edge
$(a_i,\Mopt(M(a_i)))$; ``forward'' and ``backward'' mean that the edge points
to a later or earlier top-choice block, respectively.}
Their reduction uses that \textit{OrderMatch} induces no ``backward'' edges and
then removes the remaining ``forward'' edges. For
\Cref{alg:reserved-outliers}, this property can fail: because the reserved
agents are arbitrary, assignments from $A\setminus A_{\mathrm{res}}$ to
$\Bin$ may cross top-choice blocks in either direction. We therefore need a
reduction that removes both left- and right-defects while keeping the reserved
outlier assignments fixed. The following proposition establishes this; its
proof is deferred to \Cref{app:order-match}.

\begin{restatable}{proposition}{cleanupproposition}
\label{prop:cleanup}
Fix a metric $d\in\rho(\sigma)$, and let $M$ and $A_{\mathrm{res}}$ be as
defined in \Cref{alg:reserved-outliers}. There exists a matching
$\widetilde M$ that agrees with $M$ on $A_{\mathrm{res}}$, satisfies
$\ctopk(M,d)\leq \ctopk(\widetilde M,d)$ for every $k\in [n]$, and whose
assignments of items in $\Bin$ to agents in $A^*_{\mathrm{in}}$ are all
internal.
\end{restatable}

With this proposition, we are now ready to prove one of our main results.
\begin{proof}[Proof of \Cref{thm:reserved-outliers}]
Fix a consistent metric $d \in \rho(\sigma)$.
Let $M$ be a matching returned by \Cref{alg:reserved-outliers}, and let
$\widetilde M$ be the matching guaranteed by \Cref{prop:cleanup}. 
We will show that for every agent $a_i \in A$, 
\begin{equation} \label{eq:pointwise-charge}
    d(a_i,\widetilde M(a_i))
    \leq
    d(a_i,\Mopt(a_i))
    +
    2d(\Mopt(\widetilde M(a_i)),\widetilde M(a_i)).
\end{equation} 
Granting the above inequality, we can prove our desired result as follows. 
Let $T$ be the set of $k$ agents with largest cost under
$\widetilde M$, and let
$T'=\{\Mopt(\widetilde M(a_i)):a_i\in T\}$. Since the composition of the two
perfect matchings is a bijection, we have $|T'|=k$.
Summing \eqref{eq:pointwise-charge} over $T$ gives
\[
    \ctopk(\widetilde M,d)
    \leq
    \sum_{a_i\in T} d(a_i,\Mopt(a_i))
    +
    2\sum_{a_j\in T'}d(a_j,\Mopt(a_j))
    \leq 3\ctopk(\Mopt,d).
\]
This, along with \Cref{prop:cleanup}, implies
$\ctopk(M,d)\leq 3\ctopk(\Mopt,d)$, and hence
$\distortion_k(M,d)\leq 3$. Since $d\in\rho(\sigma)$ was arbitrary,
$\distortion_k(M,\sigma)\leq 3$.

To establish \eqref{eq:pointwise-charge}, fix an agent $a_i$ and let
$a_j$ be the agent satisfying $\Mopt(a_j)=\widetilde M(a_i)$. We consider
four cases based on the locations of $a_i$ and $a_j$, as in
\citet{Filos-Ratsikas:2025aa}.

\paragraph{Case 1: $a_i,a_j\in A^*_{\mathrm{in}}$.}
Since $a_j\in A^*_{\mathrm{in}}$, the item $\widetilde M(a_i)$ lies in
$\Bin$.
By \Cref{prop:cleanup}, the assignment of $\widetilde M(a_i)$ to $a_i$ is
internal and thus, $a_i$ and $a_j$ have the same top item, say $g$.
It follows that 
\begin{align*}
    d(a_i,\widetilde M(a_i))
    &\leq d(a_i,\top_{a_i}(B))+d(a_j,\top_{a_j}(B)) 
    +d(a_j,\Mopt(a_j)) \\
    &\leq d(a_i,\Mopt(a_i))+2d(a_j,\Mopt(a_j)).
\end{align*}

\paragraph{Case 2: $a_i,a_j\notin A^*_{\mathrm{in}}$.}
By definition of the set $A^*_{\mathrm{in}}$, we have
$\Mopt(a_i),\Mopt(a_j)\in\Bout$. Thus
$b_\ell\succ_{a_i}\Mopt(a_i)$ and
$b_\ell\succ_{a_j}\Mopt(a_j)$, so
\begin{align*}
    d(a_i,\widetilde M(a_i))
    &\leq d(a_i,b_\ell)+d(a_j,b_\ell)+d(a_j,\Mopt(a_j)) \\
    &\leq d(a_i,\Mopt(a_i))+2d(a_j,\Mopt(a_j)).
\end{align*}

\noindent The next two cases use the following simple consequence of the
definition of the items $b_\ell$ and $b_r$.

\begin{claim}\label{claim:exterior-dominance}
If $b>b_r$, then $d(a,b_r)\leq d(a,b)$ for every agent $a\in A$.
Similarly, if $b<b_\ell$, then $d(a,b_\ell)\leq d(a,b)$ for every agent
$a\in A$.
\end{claim}
\begin{proof}
We prove the first statement; the second is analogous. If $a\leq b_r<b$, then
the claim is immediate. If $a>b_r$, then $\top_a(B)\leq b_r$, since $b_r$ is
the rightmost top-choice item. Thus $\top_a(B)\leq b_r<a$, which implies
$d(a,b_r)\leq d(a,\top_a(B))\leq d(a,b)$.
\end{proof}

\paragraph{Case 3: $a_i\notin A^*_{\mathrm{in}}$ and
$a_j\in A^*_{\mathrm{in}}$.}
By definition of $A^*_{\mathrm{in}}$, we have $\Mopt(a_i)\in\Bout$.
Moreover, since $a_j\in A^*_{\mathrm{in}}$ and
$\Mopt(a_j)=\widetilde M(a_i)$, we have $\widetilde M(a_i)\in\Bin$.
Assume first that $\Mopt(a_i)<b_\ell$; the case $b_r<\Mopt(a_i)$ is
symmetric. By \Cref{lem:outlier-geometry},
$\Mopt(a_i)\leq \widetilde M(a_i)$.
If $\widetilde M(a_i)\leq b_r$, then
$\Mopt(a_i)\leq \widetilde M(a_i)\leq b_r$. Since $\Mopt(a_i)$ is an
outlier, $\Mopt(a_i)\prec_{a_i} b_r$. Together with
$\Mopt(a_i)\leq \widetilde M(a_i)\leq b_r$, this gives
$\widetilde M(a_i)\succ_{a_i}\Mopt(a_i)$. Hence
\begin{align*}
    d(a_i,\widetilde M(a_i))
    &\leq d(a_i,\Mopt(a_i)) \\
    &\leq d(a_i,\Mopt(a_i))+2d(a_j,\Mopt(a_j)).
\end{align*}

If instead $\widetilde M(a_i)>b_r$, then
$d(a_j,b_r)\leq d(a_j,\Mopt(a_j))$ by
\Cref{claim:exterior-dominance}, while
$d(a_i,b_r)\leq d(a_i,\Mopt(a_i))$ because
$\Mopt(a_i)\prec_{a_i} b_r$. Hence
\begin{align*}
    d(a_i,\widetilde M(a_i))
    &\leq d(a_i,b_r)+d(a_j,b_r)+d(a_j,\Mopt(a_j)) \\
    &\leq d(a_i,\Mopt(a_i))+2d(a_j,\Mopt(a_j)).
\end{align*}

\paragraph{Case 4: $a_i\in A^*_{\mathrm{in}}$ and
$a_j\notin A^*_{\mathrm{in}}$.}
By definition of $A^*_{\mathrm{in}}$, $\Mopt(a_i)\in\Bin$ and
$\Mopt(a_j)\in\Bout$. Assume first that $\Mopt(a_j)<b_\ell$; the case
$b_r<\Mopt(a_j)$ is symmetric.
If $d(a_i,b_r)\leq d(a_i,\Mopt(a_i))$, then using that
$\Mopt(a_j)\prec_{a_j} b_r$ gives
\begin{align*}
    d(a_i,\widetilde M(a_i))
    &\leq d(a_i,b_r)+d(a_j,b_r)+d(a_j,\Mopt(a_j)) \\
    &\leq d(a_i,\Mopt(a_i))+2d(a_j,\Mopt(a_j)).
\end{align*}
Otherwise, $d(a_i,\Mopt(a_i))<d(a_i,b_r)$; suggestively,
$\Mopt(a_i)\succ_{a_i} b_r$. In this case $\Mopt(a_i)\leq b_r$ by
\Cref{claim:exterior-dominance}. Also
$\Mopt(a_j)\leq \Mopt(a_i)$ by \Cref{lem:outlier-geometry}. Thus
$\Mopt(a_j)\leq \Mopt(a_i)\leq b_r$. Since
$\Mopt(a_j)\prec_{a_j} b_r$, we have
$\Mopt(a_i)\succ_{a_j}\Mopt(a_j)$. Therefore, we derive that 
\begin{align*}
    d(a_i,\widetilde M(a_i))
    &\leq d(a_i,\Mopt(a_i))+d(a_j,\Mopt(a_i))+d(a_j,\Mopt(a_j)) \\
    &\leq d(a_i,\Mopt(a_i))+2d(a_j,\Mopt(a_j)). \qedhere
\end{align*}

\end{proof}

\section{Matching with Predictions}\label{sec:predictions}
We now consider the setting where an ordinal matching mechanism can be augmented
by a prediction. While there are several potential kinds of 
predictions that one could consider, we take a prediction to be a matching denoted by $\widehat{M}$. This design choice is in line with prior work on learning-augmented mechanism design \citep{Agrawal:2022aa}, where predictions are often taken to be an aggregate of information in the outcome space. In \Cref{app:other-predictions}, we briefly discuss how other prediction settings (e.g., locations of agents, items, etc.) would affect our results.

Given a predicted matching $\widehat{M}$, we propose a simple refinement of \Cref{alg:reserved-outliers}. Our augmented \textit{OrderMatch} mechanism uses the prediction in two ways to create a matching. First, outlier items, i.e., those in $\Bout$, will be reserved for the agents indicated by the prediction $\widehat{M}$.
Second, the
prediction $\widehat{M}$ is used to convert a weak order of the remaining agents to a total order over the remaining agents. The formal description of our mechanism is given in \Cref{alg:prediction-order-match}.

\begin{center}
\myalg{alg:prediction-order-match}{OrderMatch with Predictions}{
\textbf{Input:} A preference table $\sigma$ and predicted matching $\widehat{M}$.
\begin{enumerate}[noitemsep,leftmargin=*]
    \item Compute $\Bin$, $\Bout$, the order $\pi_B$ of $\Bin$, and the weak
    agent order $\pi_A$.
    \item Let $A_{\mathrm{res}}=\{a\in A:\widehat M(a)\in\Bout\}$. For every
    $a\in A_{\mathrm{res}}$, set $M_1(a)=\widehat M(a)$.
    \item Refine the weak order $\pi_A$ on $A\setminus A_{\mathrm{res}}$ as
    follows: if two agents have different top items, use their order in
    $\pi_A$; if two agents have the same top item, order them according to the
    order of their predicted items $\widehat M(a)$ in $\pi_B$.
    \item Let $M_2$ be the assortative matching between
    $A\setminus A_{\mathrm{res}}$ and $\Bin$ according to this
    prediction-refined agent order and the item order $\pi_B$.
    \item Return $M=M_1\cup M_2$.
\end{enumerate}
}
\end{center}
Our augmented mechanism uses the predicted matching to overcome a few known limitations of the standard \textit{OrderMatch} mechanism. One shortcoming of the original mechanism was its inability to correctly situate the outlier items in $\Bout$. 
These items can be assigned arbitrarily to get a distortion of $3$, but an improvement would seemingly require one to determine whether the outlier is to the left of $b_\ell$ or to the right of $b_r$. Another shortcoming of the original mechanism is its reliance on the weak ordering of agents deduced from $\Bin$. If several agents had the same top item, then the mechanism could incorrectly order these agents and thus may fail to assortatively match agents and items in the way that $M^*$ does. By offloading both the assignment of outlier items and the ordering of the agents to the prediction $\widehat{M}$, \Cref{alg:prediction-order-match} addresses these two issues. In the next section, we show that this approach leads to improved distortion guarantees.

\subsection{Distortion of OrderMatch with Predictions}
We are now ready to prove our main result.
Below, we show that the distortion of \Cref{alg:prediction-order-match} degrades gracefully as the prediction becomes less accurate, and is bounded above by a constant.

\begin{theorem}\label{thm:prediction-best-of-both}
For every instance $\sigma$ and predicted perfect matching $\widehat M$, let
$M$ be the matching returned by \Cref{alg:prediction-order-match} on input
$(\sigma,\widehat M)$. Then, for every $k\in[n]$, letting
$\eta_k=\distortion_k(\widehat M,\sigma)$, we have
\[
    \distortion_k(M,\sigma)\leq \min\{\eta_k,3\}.
\]
\end{theorem}

\begin{proof}

Fix an instance $\sigma$, a predicted perfect matching $\widehat M$,
$k\in[n]$, and $d\in\rho(\sigma)$. Let $M$ be the matching returned by
\Cref{alg:prediction-order-match}, and let $\Mopt$ be the optimal assortative matching
under $d$. Our approach consists of two parts. We first bound the cost with respect to the optimum matching and then with respect to the prediction.

\paragraph{Distortion with respect to $\Mopt$.}
Let $A_{\mathrm{res}}=\{a\in A:\widehat M(a)\in\Bout\}$. Since $\widehat M$ is
perfect, it matches $A_{\mathrm{res}}$ bijectively to $\Bout$. On
$A\setminus A_{\mathrm{res}}$,
\Cref{alg:prediction-order-match} assortatively matches agents to $\Bin$
according to a refinement of the recovered weak agent order. Hence $M$ can be
obtained from \Cref{alg:reserved-outliers} by reserving the items in $\Bout$ for
$A_{\mathrm{res}}$ and using $\pi_A$ and $\pi_B$ on the remaining agents and
items. By
\Cref{thm:reserved-outliers}, $\ctopk(M,d)\leq 3\ctopk(\Mopt,d)$, so
$\distortion_k(M,d)\leq 3$.

\paragraph{Distortion with respect to $\widehat{M}$.}
The matchings $M$ and $\widehat M$ agree on every agent in
$A_{\mathrm{res}}$, so we focus on comparing the costs of agents in
$A\setminus A_{\mathrm{res}}$.
Orienting the line if necessary, assume that the item order $\pi_B$ used by
the algorithm lists the items of $\Bin$ as $b_1,b_2, \ldots,b_m$ from left to right.
List the agents of $A\setminus A_{\mathrm{res}}$ as $a_1,a_2, \ldots,a_m$ so that
$\widehat M(a_i)=b_i$ for every $i\in[m]$.

Recall that \Cref{alg:prediction-order-match} orders agents with different top items using $\pi_A$ and breaks ties among agents with the same top item using the order of their predicted items.
We will reason about a sorting algorithm that starts with the initial sequence of agents $a_1, a_2, \dots, a_m$ and terminates with the refined sequence of agents used to compute $M$. This sorting procedure implicitly defines a sequence of matchings $M_0, M_1, \dots, M_T$ where $M_0 = \widehat M$ and $M_T = M$.
It suffices to show that each update in this sorting procedure weakly
decreases the top-$k$ cost. In other words, it suffices to show that for every
$t\in\{0,\ldots,T-1\}$,
\begin{equation}\label{eq:prediction-swap-invariant}
    \ctopk(M_{t+1},d)\leq \ctopk(M_t,d).
\end{equation}
Indeed, applying \eqref{eq:prediction-swap-invariant} inductively gives
\[
    \ctopk(M,d)=\ctopk(M_T,d)\leq \ctopk(M_0,d)=\ctopk(\widehat M,d),
\]
and therefore
$\distortion_k(M,d)\leq \distortion_k(\widehat M,d)$.

We now describe the sorting procedure and verify
\eqref{eq:prediction-swap-invariant}. Throughout this procedure, the
assignments of agents in $A_{\mathrm{res}}$ remain fixed. At step $t$, the current sequence
of agents is
\[
    M_t(b_1),M_t(b_2),\ldots,M_t(b_m).
\]
Suppose there exists an index $j\in[m-1]$ such that the agents
$x=M_t(b_j)$ and $y=M_t(b_{j+1})$ have different top items and $y$ precedes
$x$ in the weak order $\pi_A$. We obtain $M_{t+1}$ from $M_t$ by swapping the
assignments of $x$ and $y$. Equivalently, we swap these two adjacent agents in
the sequence. As long as such an index $j$ exists, we repeat this step. Since
each swap decreases the number of pairs of agents with different top items
that are reversed with respect to $\pi_A$, the process terminates after
finitely many steps. Moreover, since agents with the same top item are never
swapped, the final sequence is ordered according to $\pi_A$, with agents with
the same top item ordered according to their predicted items. Thus the
terminal matching is exactly $M$.

It remains to check \eqref{eq:prediction-swap-invariant} for each adjacent
swap. Consider a swap at index $j$ between agents 
$x=M_t(b_j)$ and $y=M_t(b_{j+1})$. Since $y$ precedes $x$ in $\pi_A$ and the
agents have different top items, \Cref{prop:weak-agent-order} implies that
$y<x$ in the true line order. On the other hand, $b_j<b_{j+1}$. Before the
swap, the relevant assignments are $x\mapsto b_j$ and
$y\mapsto b_{j+1}$; after the swap, they are $y\mapsto b_j$ and
$x\mapsto b_{j+1}$. Thus the $(y,x)$-swap is an uncrossing, so
\Cref{lem:uncrossing} gives exactly
\eqref{eq:prediction-swap-invariant}.

\medskip 
Combining the two bounds, we have
$\distortion_k(M,d)\leq\min\{\distortion_k(\widehat M,d),3\}$ for all $d\in\rho(\sigma)$. Taking suprema
over $d$ gives
$\distortion_k(M,\sigma)\leq\min\{\distortion_k(\widehat M,\sigma),3\}
=\min\{\eta_k,3\}$.  \qedhere

\end{proof}

Our main theorem also translates to a best of both worlds guarantee with respect to consistency and robustness. By achieving $1$-consistency and $3$-robustness, we resolve an open question of \citet{Filos-Ratsikas:2025aa} regarding the possibility of obtaining a consistency strictly better than $3$ in metric matching on the line.

\begin{corollary}[Best of both worlds] \label{cor:prediction-best-of-both}
For every $k \in [n]$, the mechanism $F$ defined in \Cref{alg:prediction-order-match} satisfies $\consistency_k(F) = 1$ and $\robustness_k(F) = 3$.
\end{corollary}
\begin{proof}
For any fixed instance $\sigma$ and prediction $\widehat{M}$, the proof of \Cref{thm:prediction-best-of-both} shows the matching $M$ returned by $F(\sigma, \widehat{M})$ satisfies 
$\distortion_k(M,d)\leq 3$ for every consistent metric $d$. 
This implies $3$-robustness.
The same proof shows that if the prediction $\widehat{M}$ is optimal for the realized metric $d$, then $\distortion_k(M,d)\leq \distortion_k(\widehat M,d)=1$. This implies $1$-consistency.
The reverse inequality for consistency follows since distortion is always at
least $1$. The reverse inequality for robustness follows from the known lower
bound of $3$ for deterministic ordinal mechanisms on the line
\citep{Filos-Ratsikas:2025aa,Caragiannis:2024ab}: fixing any
prediction $\widehat M_0$ turns $F(\cdot,\widehat M_0)$ into a deterministic
ordinal mechanism, so the robustness of $F$ is at least $3$.
\end{proof}

\section{Discussion}
In this work, we give a learning-augmented matching mechanism with a
best-of-both-worlds distortion guarantee on the line, answering an open
question of \citet{Filos-Ratsikas:2025aa}. We view our result as a proof of
concept for applying the learning-augmented framework to
metric matching problems with ordinal preferences. 

It would be interesting to understand the tradeoffs between consistency and
robustness for learning-augmented matching in more general metric spaces.
Before doing so, however, the prediction-free problem itself still needs to be
better understood: closing the gap between the best known
upper and lower bounds on distortion, $\widetilde{O}(n)$ and
$\Omega(\log n)$, respectively, remains a central open question.
Determining the right worst-case guarantee is a natural prerequisite for asking
what predictions can add.

At the same time, our analysis suggests a design principle for more general
metric spaces. Our mechanism follows a familiar \emph{bounding-box} pattern from
learning-augmented algorithms: identify a robust family of plausible solutions,
then use the prediction to select among them. Here,
\Cref{alg:reserved-outliers} supplies such a family of low-distortion matchings.
This suggests seeking ordinal mechanisms for general metrics that produce viable
families rather than a single output. In voting, \citet{Berger:2024aa} use the
$(p,q)$-veto core \citep{Kizilkaya:2023ab} to produce viable candidate sets. 
A natural question is whether RepMatch and its variants
\citep{Anari:2023aa,Hastings:2025aa} can similarly produce viable matching
families, ideally ones containing an optimal matching when the prediction is
accurate.

\bibliography{mybib}
\bibliographystyle{apalike}

\appendix

\section{AI Disclosure}
The authors used GPT-5.5 Extra High in Codex extensively throughout the
research process for brainstorming, generating and refining proof sketches,
organizing arguments, and streamlining the presentation.
Notably, the model assisted in writing up the swap arguments used in
\Cref{prop:cleanup} and \Cref{thm:prediction-best-of-both} from
author-provided directions and proof ideas, and in creating the illustrative
figure accompanying \Cref{prop:cleanup}. It also helped identify and
formalize adaptations of the lower-bound construction of
\citet{Filos-Ratsikas:2025aa} for alternative prediction models discussed 
in \Cref{app:other-predictions}. The
authors accept full responsibility for the correctness of the work.

\section{Missing Proofs from \texorpdfstring{\Cref{sec:prelims}}{Preliminaries}}\label{app:prelims}

\uncrossinglemma*
\begin{proof}[Proof of \Cref{lem:uncrossing}]
Note that only the individual costs for agents $a_i$ and $a_j$ change.
Consider the two-agent subinstance with agents $a_i,a_j$ and items
$M(a_j),M(a_i)$. Since $a_i<a_j$ and $M(a_j)<M(a_i)$, the swap changes the
crossed matching into the assortative matching on this subinstance. By
\Cref{lem:assortative}, this weakly decreases both the top-$1$ and top-$2$
costs of these two agents; that is,
\begin{align}
    \max\{d(a_i,M(a_j)),d(a_j,M(a_i))\}
    &\leq \max\{d(a_i,M(a_i)),d(a_j,M(a_j))\}, \label{eq:crossing-max}\\
    d(a_i,M(a_j))+d(a_j,M(a_i))
    &\leq d(a_i,M(a_i))+d(a_j,M(a_j)). \label{eq:crossing-sum}
\end{align}
Now let $S$ be a set of $k$ agents attaining $\ctopk(M',d)$. If $S$ contains neither affected agent, take $T=S$. Otherwise, applying
\eqref{eq:crossing-max} if $|S\cap\{a_i,a_j\}|=1$ and
\eqref{eq:crossing-sum} if $|S\cap\{a_i,a_j\}|=2$, and replacing the affected
agent in $S$ in the first case if necessary, we obtain a set $T$ of $k$ agents
such that
\[
    \ctopk(M',d)
    =
    \sum_{a\in S} d(a,M'(a))
    \leq
    \sum_{a\in T} d(a,M(a))
    \leq
    \ctopk(M,d),
\]
which implies our desired result.
\end{proof}

\section{Missing Proofs from \texorpdfstring{\Cref{sec:line}}{Matching without Predictions}}\label{app:order-match}
\weakitemorderlemma*
\begin{proof}[Proof of \Cref{lem:weak-item-order}]
We present an argument similar in spirit to order-recovery procedures
of \citep{Elkind:2014aa, Babashah:2025aa}, but we give a direct proof tailored
to the set $\Bin$.

First, note that if $b_\ell=b_r$, then $B_+=\{b_\ell\}$, so
$\Bin=\{b_\ell\}$ and there is nothing to prove. Hence, we assume
$b_\ell < b_r$.
The high-level idea is to first partition the items in $\Bin$ based on their
locations relative to $b_\ell$ and $b_r$ on the line and then sort each of these
sets separately. In what follows, let $a_\ell$ and $a_r$ be agents such that
$\top_{a_\ell}(B)=b_\ell$ and $\top_{a_r}(B)=b_r$.

From the preference profile, we form the sets
\[
    L=\{b\in\Bin: b\prec_a b_\ell \text{ for every } a\in A\},
    \qquad
    R=\{b\in\Bin: b_r\succ_a b \text{ for every } a\in A\},
\]
and let $M=\Bin\setminus(L\cup R)$.
\begin{claim}
\label{claim:partition}
The sets $L$, $M$, and $R$ partition $\Bin$ into the items lying strictly left
of $b_\ell$, in the interval $[b_\ell,b_r]$, and strictly right of $b_r$,
respectively.
\end{claim}
\begin{proof}
First, we show that if $b\in\Bin$ and $b<b_\ell$, then $b\in L$.
Suppose for contradiction that $b\succ_a b_\ell$ for some agent $a$. Then it
must be that $a<b_\ell$. If $\top_a(B)=b_\ell$, this contradicts
$b\succ_a b_\ell$. Otherwise, since no item to the right of $b_\ell$ can be
ranked above $b_\ell$ by an agent to the left of $b_\ell$, we must have
$\top_a(B)<b_\ell$, contradicting the choice of $b_\ell$ as the leftmost
top-choice item. Hence $b\prec_a b_\ell$ for every agent $a$, which implies
$b\in L$.
Since $b\in\Bin$ and no agent has $b\succeq_a b_\ell$, some agent must have
$b\succ_a b_r$. Therefore $b\notin R$, so $b\in L\setminus R$ and, by
definition, $b\notin M$.

An analogous argument shows that if $b\in\Bin$ and $b>b_r$, then
$b\in R\setminus L$.

Now, consider an item $b\in\Bin$ such that $b_\ell<b<b_r$. Since
$\top_{a_r}(B)=b_r$, the agent $a_r$ lies to the right of $b$, and therefore
$b\succ_{a_r} b_\ell$. Thus $b\notin L$. Similarly, $a_\ell$ lies to the left
of $b$, so $b\succ_{a_\ell} b_r$ and $b\notin R$. Finally, the endpoints
$b_\ell$ and $b_r$ are not in $L\cup R$. Hence every item of $\Bin$ in
$[b_\ell,b_r]$ belongs to $M$, and the claimed partition follows.
\end{proof}

Now, let $\pi_L$ be the reverse of agent $a_r$'s preference
list restricted to $L$. Let $\pi_M$ and $\pi_R$ be agent $a_\ell$'s preference
list restricted to $M$ and $R$, respectively. Our algorithm returns the
concatenation $\pi_B=\pi_L\circ\pi_M\circ\pi_R$.
If the initial labels of the two extreme items $b_\ell$ and $b_r$ are reversed,
the same procedure returns the reversal of this order.

\begin{claim}
\label{claim:sorted-items}
The order $\pi_B$ is the true left-to-right ordering of $\Bin$.
\end{claim}
\begin{proof}
By \Cref{claim:partition}, every item in $L$ lies to the left of every item in
$M$, and every item in $M$ lies to the left of every item in $R$. It remains
to check the order within each block.

Since $\top_{a_\ell}(B)=b_\ell$, no item lies strictly between $a_\ell$ and
$b_\ell$; otherwise that item would be ranked above $b_\ell$ by $a_\ell$. If
$a_\ell\leq b_\ell$, then every item of $M\cup R$ lies to the right of
$a_\ell$, so for any two items $b_i,b_j\in M\cup R$ with $b_i<b_j$, we have
$b_i\succ_{a_\ell} b_j$. If
$b_\ell<a_\ell$, then $b_\ell$ is first in $a_\ell$'s preference list, and
every other item of $M\cup R$ lies to the right of $a_\ell$. Again, for any
two such remaining items $b_i<b_j$, we have $b_i\succ_{a_\ell} b_j$. Thus
$a_\ell$'s preference list orders $M\cup R$ from left to right, and therefore
orders both $M$ and $R$ from left to right.

Similarly, $a_r$'s preference list orders $L$ from right to left, so the
reverse of this restricted list orders $L$ from left to right. It follows that
$\pi_B$ is the true left-to-right ordering of $\Bin$.
\end{proof}

Finally, note that the partitioning and sorting steps clearly run in polynomial
time.
\end{proof}

\outliergeometrylemma*
\begin{proof}[Proof of \Cref{lem:outlier-geometry}]
If $\bout$ lies in the closed interval between $b_\ell$ and $b_r$,
then take any agent $a\in A$. If $a\leq \bout$, then $\bout$ lies between
$a$ and $b_r$, so $\bout\succ_a b_r$ unless $\bout=b_r$, in which case
$\bout\succeq_a b_r$. If instead $\bout<a$, then $\bout$ lies between
$a$ and $b_\ell$, so $\bout\succ_a b_\ell$ unless $\bout=b_\ell$, in which
case $\bout\succeq_a b_\ell$. In either case, $\bout\in\Bin$,
contradicting $\bout\in\Bout$. Thus
every outlier lies outside the interval $[b_\ell, b_r]$.

Now suppose $\bout<b_\ell$ and, for contradiction, that some $b\in\Bin$
satisfies $b<\bout$. Then we have $b<\bout<b_\ell\leq b_r$. Since
$b\in\Bin$, there exists an agent $a\in A$ and an extreme item
$c\in\{b_\ell,b_r\}$ such that $b\succeq_a c$. If $\bout<a$, then
$\bout\succ_a b\succeq_a c$, and therefore $\bout\in\Bin$. If instead
$a\leq \bout$, then $\bout\succeq_a b_\ell$, and again $\bout\in\Bin$. Both
cases contradict $\bout\in\Bout$. Therefore $\bout\leq b$ for every
$b\in\Bin$. This establishes property (1). The proof of property (2) is
analogous.
\end{proof}

\weakagentorderlemma*
\begin{proof}[Proof of \Cref{prop:weak-agent-order}]
Construct $\pi_A$ by ordering agents according to the position of
$\top_a(B)$ in $\pi_B$, tying agents with the same top item. Since
$\top_a(B)\in B_+\subseteq\Bin$ for every agent $a$, this procedure is
well-defined. Property (1) of the lemma immediately follows.
It is also straightforward to verify that this sorting procedure can be done
efficiently.

For property (2), suppose $a_i$ strictly precedes $a_j$ in $\pi_A$. Then
$\top_{a_i}(B)<\top_{a_j}(B)$ in the true left-to-right order. Since each
agent is closest to its top item, consistency gives
\[
    d(a_i,\top_{a_i}(B))\leq d(a_i,\top_{a_j}(B))
    \quad\text{and}\quad
    d(a_j,\top_{a_j}(B))\leq d(a_j,\top_{a_i}(B)).
\]
On the line, these inequalities place $a_i$ weakly to the left of the
midpoint between the two top items and $a_j$ weakly to the right of that same
midpoint. Thus $a_i\leq a_j$. Since agents have distinct locations, we in fact
have $a_i<a_j$, as desired.
\end{proof}

\cleanupproposition*

\begin{proof}[Proof of \Cref{prop:cleanup}]
Orienting the line if necessary, assume that the recovered order $\pi_B$ from
\Cref{lem:weak-item-order} is the true left-to-right order of $\Bin$.
The agents matched to $\Bin$ by \Cref{alg:reserved-outliers} are exactly
$A\setminus A_{\mathrm{res}}$.
Partition the agents $A^*_{\mathrm{in}}$ into blocks according to their
top-choice item, so agents in the same block have the same top choice. Let
$A_1,\ldots,A_r$ be these blocks, ordered according to the recovered item order
from \Cref{lem:weak-item-order}, oriented consistently with the true
left-to-right order. By \Cref{prop:weak-agent-order}, this block order agrees
with the true left-to-right order of the agents.

For each $s\in[r]$, set $B_s=\Mopt(A_s)$. The blocks
$B_1,\ldots,B_r$ partition $\Bin$, since every item in $\Bin$ is matched by
$\Mopt$ to a unique agent in $A^*_{\mathrm{in}}$; in particular,
$|B_s|=|A_s|$. These item blocks inherit the same left-to-right order: if
$s<t$, then every agent in $A_s$ lies to the left of every agent in $A_t$ by
\Cref{prop:weak-agent-order}, and assortativity of $\Mopt$ implies that every
item in $B_s$ lies to the left of every item in $B_t$.

We will transform $M$ by swapping only agents in $A\setminus A_{\mathrm{res}}$; the
assignments of agents in $A_{\mathrm{res}}$ remain fixed throughout. The goal is
to remove cross-block assignments. In the final matching, if an agent in
$A^*_{\mathrm{in}}\cap (A\setminus A_{\mathrm{res}})$ belongs to block $A_s$,
then it should receive an item from $B_s$.

\paragraph{Left and right defects.}
For a matching $N$ that agrees with $M$ on $A_{\mathrm{res}}$, define
\begin{align*}
    D_s^+(N)
    &=
    \{a\in A_s\cap (A\setminus A_{\mathrm{res}}):
        N(a)\in B_t\text{ for some }t>s\}, \\
    D_s^-(N)
    &=
    \{a\in A_s\cap (A\setminus A_{\mathrm{res}}):
        N(a)\in B_t\text{ for some }t<s\}.
\end{align*}
These are the right- and left-defects in block $A_s$, respectively. They record
the agents in $A_s\cap (A\setminus A_{\mathrm{res}})$ whose current item under
$N$ comes from the wrong item block: $D_s^+(N)$ records those receiving an item
from a block to the right of $B_s$, and $D_s^-(N)$ records those receiving an
item from a block to the left of $B_s$. If instead such an agent $a$ has
$N(a)\in B_s$, then the assignment is internal by definition, since
$B_s=\Mopt(A_s)$ is the item block paired with $A_s$ by $\Mopt$. This definition
only compares blocks, so it does not require an ordering of agents inside a
single block $A_s$.

\paragraph{Exterior agents and swap partners.}
For ease of exposition, call the agents in
$(A\setminus A_{\mathrm{res}})\setminus A^*_{\mathrm{in}}$ \textit{exterior}.
These are the agents matched to $\Bin$ by $M$ but matched to $\Bout$ by
$\Mopt$. By assortativity of $\Mopt$ and \Cref{lem:outlier-geometry}, every
exterior agent lies either to the left of all blocks $A_1,\ldots,A_r$ or to
their right; call these left and right exterior agents, respectively. For a
block $A_s$, let $L_s$ be the left exterior agents together with
$\bigcup_{u<s}(A_u\cap (A\setminus A_{\mathrm{res}}))$, and let $R_s$ be the
right exterior agents together with
$\bigcup_{u>s}(A_u\cap (A\setminus A_{\mathrm{res}}))$. Thus $L_s$ and $R_s$
are the agents on the two sides of $A_s$ whose assignments we are still allowed
to change.

We will only use the following two kinds of swaps:
\begin{enumerate}[label=(\roman*),ref=(\roman*), noitemsep]
    \item \label{swap:right-defect}
    To fix a right-defect in $A_s$, swap an agent in $D_s^+(N)$ with an agent
    to the left of $A_s$ that currently holds an item from $B_s$.
    \item \label{swap:left-defect}
    To fix a left-defect in $A_s$, swap an agent in $D_s^-(N)$ with an agent to
    the right of $A_s$ that currently holds an item from $B_s$.
\end{enumerate}
For a matching $N$, the possible partners for these two swaps are
\[
    P_s(N)=\{a\in L_s:N(a)\in B_s\}
    \quad\text{and}\quad
    Q_s(N)=\{a\in R_s:N(a)\in B_s\}.
\]
That is, $P_s(N)$ is the set of possible partners for swaps of
type~\ref{swap:right-defect}, while $Q_s(N)$ is the set of possible partners
for swaps of type~\ref{swap:left-defect}.

\begin{figure}[H]
    \centering
    \begin{tikzpicture}[
        x=1cm,y=1cm,
        block/.style={draw,rounded corners=4pt,minimum height=0.65cm,inner sep=3pt},
        every node/.style={font=\scriptsize}
    ]
        \node[block,minimum width=1.15cm,label=above:$B_s$] (bsbox) at (1.95,1.45) {};
        \node[block,minimum width=1.15cm,label=above:$B_{s+1}$] at (3.15,1.45) {};
        \node[block,minimum width=1.15cm,label=above:$B_t$] (btbox) at (5.05,1.45) {};
        \node (bs) at (bsbox.center) {$b_s$};
        \node (bt) at (btbox.center) {$b_t$};
        \node at (4.1,1.45) {$\cdots$};

        \node[block,minimum width=2.25cm,label=below:$L_s$] (ajbox) at (0.0,0.05) {};
        \node[block,minimum width=1.15cm,label=below:$A_s$] (aibox) at (1.95,0.05) {};
        \node[block,minimum width=1.15cm,label=below:$A_{s+1}$] at (3.15,0.05) {};
        \node[block,minimum width=1.15cm,label=below:$A_t$] at (5.05,0.05) {};
        \node[block,minimum width=1.6cm,label=below:$R_t$] at (6.65,0.05) {};

        \node (aj) at (ajbox.center) {$P_s(N)\ni a_j$};
        \node (ai) at (aibox.center) {$a_i$};
        \node at (4.1,0.05) {$\cdots$};

        \draw[black,densely dotted,line width=0.8pt] (aj.north east) -- (bs.south);
        \draw[black,line width=0.8pt] (ai.north) -- (bt.south);
    \end{tikzpicture}
    \caption{A type~\ref{swap:right-defect} swap. The assignments shown are
    edges of the current matching $N$; initially, $N=M$. The solid edge is the
    right-defect assignment $N(a_i)\in B_t$ for some $t>s$, while the dotted
    edge shows that the partner $a_j\in P_s(N)$ currently holds an item from
    $B_s$.
    Type~\ref{swap:left-defect} is symmetric.}
    \label{fig:right-defect-swap}
\end{figure}
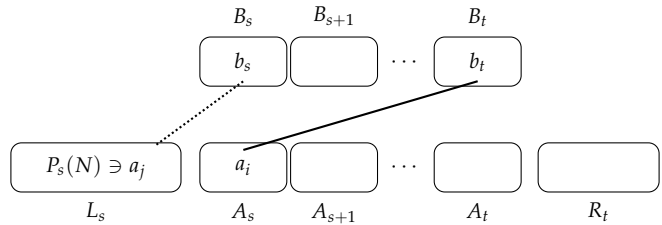

\begin{claim}[Availability of swap partners] \label{claim:block-counts}
For the initial matching $M$ returned by \Cref{alg:reserved-outliers}, every
block $A_s$ satisfies
\begin{equation}\label{eq:partner-availability}
\begin{aligned}
    |P_s(M)|\geq |D_s^+(M)|
    \qquad\text{and}\qquad
    |Q_s(M)|\geq |D_s^-(M)|.
\end{aligned}
\end{equation}
Moreover, if the inequalities in \eqref{eq:partner-availability} hold for a
matching $N$, then they still hold after any swap of
type~\ref{swap:right-defect}; such a swap decreases the total number of
right-defects by one and creates no left-defect. If the second inequality in
\eqref{eq:partner-availability} holds for $N$, then it still holds after any
swap of type~\ref{swap:left-defect}; such a swap decreases the total number of
left-defects by one and creates no right-defect.
\end{claim}
\begin{proof}
We first prove the two inequalities for the initial matching $M$. For this
part, use the ordered lists in \Cref{alg:reserved-outliers}: the agents in
$A\setminus A_{\mathrm{res}}$ are sorted by top-choice block, while ties inside
a block are arbitrary, and the items in $\Bin$ are sorted from left to right.
Thus the interval $B_s$ appears after $B_1,\ldots,B_{s-1}$ and before
$B_{s+1},\ldots,B_r$.

Fix a block $A_s$. After deleting exterior agents from the ordered agent list,
the remaining agents of $A_s$ form one block, with earlier and later top-choice
blocks on the corresponding sides. Now return to the full ordered list. If
$q=|D_s^+(M)|$ agents of $A_s\cap(A\setminus A_{\mathrm{res}})$ are matched to
items strictly to the right of $B_s$, then at least $q$ positions inside $B_s$
must be occupied by agents that appear on the left side of $A_s$: otherwise the
assignments of $A_s$ could not extend past the right end of $B_s$. These agents
lie in $L_s$ and hold distinct items of $B_s$, so they give $q$ distinct
elements of $P_s(M)$. Hence $|P_s(M)|\geq |D_s^+(M)|$. Reading the same ordered
lists from right to left gives $|Q_s(M)|\geq |D_s^-(M)|$.

Next consider a swap of type~\ref{swap:right-defect}. Thus an agent
$a_i\in D_s^+(N)$ is swapped with an agent $a_j\in P_s(N)$, where
$N(a_i)\in B_t$ for some $t>s$ and $N(a_j)\in B_s$. For block $s$, both
$|D_s^+|$ and $|P_s|$ decrease by one. If $a_j$ is in an earlier block $A_u$,
then before the swap it was already in $D_u^+(N)$ and after the swap it still
is, since it receives an item in $B_t$ with $t>s>u$; if $a_j$ is exterior, it
creates no defect. Thus the swap decreases the total number of right-defects by
one and creates no left-defect. The only other possible change in the partner
sets involves the item of $B_t$, and that item moves from one agent in $L_t$ to
another agent in $L_t$. Hence the inequalities in
\eqref{eq:partner-availability} are preserved.

The argument for a swap of type~\ref{swap:left-defect} is the mirror image, but
we spell out the relevant changes. Now an agent $a_i\in D_s^-(N)$ is swapped
with an agent $a_j\in Q_s(N)$, where $N(a_i)\in B_t$ for some $t<s$ and
$N(a_j)\in B_s$. For block $s$, both $|D_s^-|$ and $|Q_s|$ decrease by one.
If $a_j$ is in a later block $A_u$, then it was already in $D_u^-(N)$ and
remains so after receiving an item in $B_t$ with $t<s<u$; if $a_j$ is exterior,
it creates no defect. Thus the swap decreases the total number of left-defects
by one and creates no right-defect. The only other possible change in the
partner sets involves the item of $B_t$, and that item moves from one agent in
$R_t$ to another agent in $R_t$. Hence the second inequality in
\eqref{eq:partner-availability} is preserved.
\end{proof}

We now use \Cref{claim:block-counts} to transform $M$ in two passes. First we
remove all right-defects; then, keeping right-defects from reappearing, we
remove all left-defects.

\paragraph{First phase: removing right-defects.}
Starting from $M$, let $N$ be the current matching. If there is a right-defect
in some block $A_s$, choose $a_i\in D_s^+(N)$. By
\Cref{claim:block-counts}, $P_s(N)$ is nonempty, so choose
$a_j\in P_s(N)$. Then $a_j$ lies to the left of $a_i$ in the true line order,
and $N(a_j)\in B_s$ lies to the left of $N(a_i)$. Thus the $(a_j,a_i)$-swap is
the reverse of an uncrossing. By \Cref{claim:block-counts}, after this swap the
two inequalities still hold, the total number of right-defects decreases by one,
and no left-defect is created. Repeating this step terminates with no
right-defects.

\paragraph{Second phase: removing left-defects.}
Now suppose no right-defect remains, and again let $N$ be the current matching.
If there is a left-defect in some block $A_s$, choose
$a_i\in D_s^-(N)$. By \Cref{claim:block-counts}, $Q_s(N)$ is nonempty, so choose
$a_j\in Q_s(N)$. Then $a_i<a_j$ in the true line order, and $N(a_i)$ lies to the
left of $N(a_j)\in B_s$. Hence the $(a_i,a_j)$-swap is the reverse of an
uncrossing. By \Cref{claim:block-counts}, the second inequality remains true,
the number of left-defects decreases by one, and no right-defect is created.
Repeating this step terminates with no left-defects. Let $\widetilde M$ be the
resulting matching. Since
$\widetilde M$ has no right- or left-defects, every assignment of an item in
$\Bin$ to an agent in $A^*_{\mathrm{in}}$ is internal.

\medskip 
Finally, note that every step in the reduction from $M$ to $\widetilde M$ is the reverse of an
uncrossing, so by \Cref{lem:uncrossing} each step weakly increases
$\ctopk(\cdot,d)$ for every $k\in[n]$. Therefore, for every $k\in[n]$:
\[
    \ctopk(M,d)
    \leq
    \ctopk(\widetilde M,d).
\]
\end{proof}

\section{Other Prediction Settings}\label{app:other-predictions}

Apart from a matching $\widehat{M}$, there are several other types of
predictions that a matching mechanism could reasonably be expected to receive.
Here, we briefly discuss the most natural alternatives and their implications.
Broadly, the predictions $\widehat{P}$ that we will consider fall into two
categories: those that predict agent/item positions, and those that predict
agent/item identities.

For any such prediction format, let $F(\sigma,\widehat{P})$ be a
deterministic mechanism that receives the ordinal profile $\sigma$ and the
prediction $\widehat{P}$. Let $\rho(\sigma,\widehat{P})$ denote the set of
metrics that are consistent both with $\sigma$ and with the prediction
$\widehat{P}$. The corresponding consistency benchmark is
\[
    \consistency^{\widehat{P}}_k(F)
    =
    \sup_{\sigma,\widehat{P}}
    \sup_{d\in\rho(\sigma,\widehat{P})}
    \distortion_k(F(\sigma,\widehat{P}),d).
\]

\paragraph{Predicted agent locations.}
A mechanism could be given information so that the relative order of agents
could be gleaned. Perhaps the strongest form of this prediction would be the
exact locations of agents on the line. In this setting, the prediction
$\widehat{P}$ is a vector indexed by agents, where $\widehat{P}(a)$ is the
location of agent $a$.

\begin{claim}\label{claim:agent-locations-insufficient}
For every $k\in[n]$, every deterministic mechanism $F$ augmented with the exact
locations of the agents satisfies
$\consistency^{\widehat{P}}_k(F)\geq 3$. Consequently, agent
locations alone cannot yield consistency better than $3$, even if no
robustness requirement is imposed.
\end{claim}
\begin{proof}
We use an indistinguishability argument, in the same spirit as the lower-bound
template of \citet{Filos-Ratsikas:2025aa}. At a high level, the hidden bit is
the side of the outlier item $b_n$. We construct two metrics, $d_L$ and
$d_R$, that have the same prediction $\widehat{P}$ and induce the same ordinal
profile. Since $F$ receives the same input in both metrics, it
must commit to the same recipient of $b_n$. We then evaluate it on the metric
in which that recipient is far from $b_n$, while the optimum matches $b_n$ to
a nearby agent.

We consider cases based on the value of $k$. In what follows, fix
$\varepsilon\in(0,1/2)$ and set $\delta=\varepsilon/(2n^2)$.

\paragraph{Case of $k=1$:}
Let $A_0=\{a_1,\ldots,a_{n-1}\}$ and
$B_0=\{b_1,\ldots,b_{n-1}\}$. Set $\widehat{P}(a_i)=(i-1)\delta$ for
$i\in[n-1]$ and $\widehat{P}(a_n)=2-\varepsilon$. In both metrics, place item
$b_i$ at
\[
    1-\varepsilon/2+(i-1)\delta
    \qquad\text{for every } i\in[n-1].
\]
The metrics differ only in the position of $b_n$: in $d_R$, place $b_n$ at
$3-\varepsilon$; in $d_L$, place $b_n$ at $-1$. For each agent, the order of
the items in $B_0$ is identical in $d_L$ and $d_R$, and the choice of
$\delta$ ensures that every item in $B_0$ is closer than $b_n$ in both
metrics. Thus $d_L$ and $d_R$ have the same predicted locations and induce the
same ordinal profile. See \Cref{fig:agent-location-k1}.

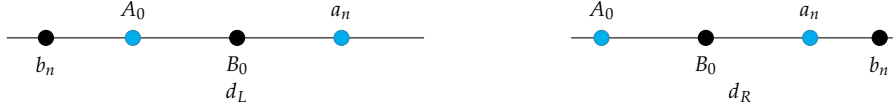
\begin{figure}[H]
    \centering
    \begin{tikzpicture}[
        x=1.15cm,y=0.9cm,
        agent/.style={circle,draw=cyan!60!black,fill=\agentcolor,inner sep=2pt},
        item/.style={circle,draw=black,fill=\itemcolor,inner sep=2pt},
        every node/.style={font=\scriptsize}
    ]
        \begin{scope}
            \draw[black] (-1.45,0) -- (3.35,0);
            \node[item,label=below:$b_n$] at (-1,0) {};
            \node[agent,label=above:$A_0$] at (0,0) {};
            \node[item,label=below:$B_0$] at (1.2,0) {};
            \node[agent,label=above:$a_n$] at (2.4,0) {};
            \node at (1.2,-0.8) {$d_L$};
        \end{scope}
        \begin{scope}[xshift=6.2cm]
            \draw[black] (-0.35,0) -- (3.35,0);
            \node[agent,label=above:$A_0$] at (0,0) {};
            \node[item,label=below:$B_0$] at (1.2,0) {};
            \node[agent,label=above:$a_n$] at (2.4,0) {};
            \node[item,label=below:$b_n$] at (3.2,0) {};
            \node at (1.6,-0.8) {$d_R$};
        \end{scope}
    \end{tikzpicture}
    \caption{The $k=1$ agent-location construction, with $A_0=\{a_1,\ldots,a_{n-1}\}$ and $B_0=\{b_1,\ldots,b_{n-1}\}$. Clusters are drawn as single nodes; in the proof their points are separated by $\delta$.}
    \label{fig:agent-location-k1}
\end{figure}

We first compute $\cost_1(\Mopt,\cdot)$ in the two metrics. In $d_R$, match
$b_n$ to $a_n$ and match $a_i$ to $b_i$ for every $i\in[n-1]$. The edge
incident to $b_n$ has cost $1$, and every other edge has cost
$1-\varepsilon/2$, so $\cost_1(\Mopt,d_R)=1$. In $d_L$, match $b_n$ to
$a_1$, match $a_i$ to $b_{i-1}$ for $i=2,\ldots,n-1$, and match $a_n$ to
$b_{n-1}$. Since $a_1$ is the closest agent to $b_n$ in $d_L$,
$\cost_1(\Mopt,d_L)=1$.

Now let $M$ be the matching returned by $F$ on this common input. If
$M(b_n)\in A_0$, evaluate $M$ on $d_R$. Then
$\cost_1(M,d_R)\geq 3-\varepsilon-(n-2)\delta$; since
$\cost_1(\Mopt,d_R)=1$, the same lower bound holds for
$\distortion_1(M,d_R)$. If $M(b_n)=a_n$, evaluate $M$ on $d_L$. Then
$\cost_1(M,d_L)\geq 3-\varepsilon$, and hence
$\distortion_1(M,d_L)\geq 3-\varepsilon$. Thus, for some
$d\in\{d_L,d_R\}$, $\distortion_1(M,d)\geq
3-\varepsilon-(n-2)\delta$.
Since $\delta=\varepsilon/(2n^2)$, taking $\varepsilon\to0$ makes this lower
bound tend to $3$. This gives the desired lower bound for this case.

\paragraph{Case of $k\geq 2$:}
Let $A_0=\{a_1,\ldots,a_{n-1}\}$ and
$B_0=\{b_2,\ldots,b_{n-1}\}$, possibly empty. Set
$\widehat{P}(a_i)=-(n-i+1)\delta$ for $i\in[n-1]$ and
$\widehat{P}(a_n)=1-\varepsilon$. The items in $B_0$ have the same
locations in both metrics: for $i=2,\ldots,n-1$, place $b_i$ at
$(i-2)\delta$. In the right metric $d_R$, place $b_1$ at
$(n-2)\delta$ and $b_n$ at $2-\varepsilon$. In the left metric $d_L$, place
$b_n$ at $-1$ and $b_1$ at $1-\varepsilon+\delta$.
The ordinal profile is: agents $a_1,\ldots,a_{n-1}$ rank $B_0$ first, then
$b_1$, then $b_n$; agent $a_n$ ranks $b_1$ first, then $B_0$, then $b_n$.
Ties inside $B_0$ follow their induced order. The near-zero cluster
comparisons are immediate in both metrics; the only one to check is, for a
left-cluster agent in $d_L$, $b_1$ versus $b_n$:
\[
    d(a_i,b_1)
    =
    1-\varepsilon+(n-i+2)\delta
    <
    1-(n-i+1)\delta
    =
    d(a_i,b_n),
\]
since $(2n-2i+3)\delta<\varepsilon$. See \Cref{fig:agent-location-kgeq2}.

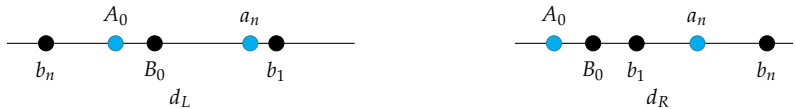
\begin{figure}[H]
    \centering
    \begin{tikzpicture}[
        x=1.15cm,y=0.9cm,
        agent/.style={circle,draw=cyan!60!black,fill=\agentcolor,inner sep=2pt},
        item/.style={circle,draw=black,fill=\itemcolor,inner sep=2pt},
        every node/.style={font=\scriptsize}
    ]
        \begin{scope}
            \draw[black] (-1.45,0) -- (2.55,0);
            \node[item,label=below:$b_n$] at (-1,0) {};
            \node[agent,label=above:$A_0$] at (-0.2,0) {};
            \node[item,label=below:$B_0$] at (0.25,0) {};
            \node[agent,label=above:$a_n$] at (1.35,0) {};
            \node[item,label=below:$b_1$] at (1.65,0) {};
            \node at (0.55,-0.8) {$d_L$};
        \end{scope}
        \begin{scope}[xshift=5.8cm]
            \draw[black] (-0.65,0) -- (2.55,0);
            \node[agent,label=above:$A_0$] at (-0.2,0) {};
            \node[item,label=below:$B_0$] at (0.25,0) {};
            \node[item,label=below:$b_1$] at (0.75,0) {};
            \node[agent,label=above:$a_n$] at (1.45,0) {};
            \node[item,label=below:$b_n$] at (2.25,0) {};
            \node at (1.0,-0.8) {$d_R$};
        \end{scope}
    \end{tikzpicture}
    \caption{The $k\geq2$ agent-location construction, with $A_0=\{a_1,\ldots,a_{n-1}\}$ and $B_0=\{b_2,\ldots,b_{n-1}\}$. The perturbation makes the left clusters distinct while preserving the hidden side of $b_n$.}
    \label{fig:agent-location-kgeq2}
\end{figure}

We first compute a useful upper bound on $\cost_k(\Mopt,\cdot)$. In $d_R$,
match $b_n$ to $a_n$ and the remaining agents assortatively to the near-zero
item cluster. The edge incident to $a_n$ has cost $1$, and every other edge
has cost $n\delta$, so $\cost_k(\Mopt,d_R)\leq 1+(k-1)n\delta$. In $d_L$,
match $b_n$ to $a_1$, match $b_1$ to $a_n$, and match the remaining agents
assortatively to $B_0$. These edge costs are at most $1-n\delta$, $\delta$,
and $(n-1)\delta$, respectively, so
$\cost_k(\Mopt,d_L)\leq 1+(k-1)n\delta$.

Now let $M$ be the matching returned by $F$ on this common input. If
$M(b_n)=a_i\neq a_n$, evaluate $M$ on $d_R$. Agent $a_i$ pays
$2-\varepsilon+(n-i+1)\delta\geq 2-\varepsilon$ for $b_n$. Since $a_n$ does
not receive $b_n$, it receives an item in the near-zero cluster and pays at
least $1-\varepsilon-(n-2)\delta$. Because $k\geq 2$,
$\cost_k(M,d_R)\geq (3-2\varepsilon)-(n-2)\delta$. If $M(b_n)=a_n$, evaluate
$M$ on $d_L$. Then $a_n$ pays $2-\varepsilon$ for $b_n$, and $b_1$ is assigned
to some agent at coordinate at most $-2\delta$, who pays at least
$1-\varepsilon$. Hence $\cost_k(M,d_L)\geq 3-2\varepsilon$, which is at least
$(3-2\varepsilon)-(n-2)\delta$. Therefore, for some $d\in\{d_L,d_R\}$,
\[
    \distortion_k(M,d)
    \geq
    \frac{3-2\varepsilon-(n-2)\delta}{1+(k-1)n\delta}.
\]
Since $\delta=\varepsilon/(2n^2)$, taking $\varepsilon\to0$ makes this lower
bound tend to $3$. This gives the desired lower bound for this case.

Together, the two cases show that the consistency of $F$ is at least $3$ for
every $k\in[n]$.
\end{proof}

\paragraph{Predicted item locations.}

One might instead hope that predictions about the item side of the market are
enough. We state the lower bound for the stronger prediction that the exact
item locations are known; it immediately implies the same lower bound when
only the item order is predicted. In this setting, the prediction
$\widehat{P}$ is a vector indexed by items, where $\widehat{P}(b)$ is the
location of item $b$.

\begin{claim}\label{claim:item-locations-insufficient}
For every $k\in[n]$, every deterministic mechanism $F$ augmented with the exact
locations of the items satisfies
$\consistency^{\widehat{P}}_k(F)\geq 3$. Consequently, the item
order alone cannot yield consistency better than $3$.
\end{claim}
\begin{proof}
The proof again uses indistinguishable realizations of the same prediction,
but now the hidden information is the identity of the agent that should
receive the far item. We fix the item locations and construct two special
agents with identical preference lists. The mechanism cannot tell which
special agent is near $b_n$, so after it commits to a recipient for $b_n$, we
choose the consistent realization in which this decision is costly. The filler
items are placed far to the right so that spreading them apart does not reveal
the hidden special-agent identity.

Fix a deterministic item-location-augmented mechanism $F$.
Agents $a_1$ and $a_2$ are special; all remaining agents, when present, form
$A_F=\{a_3,\ldots,a_n\}$. Write $B_F=\{b_2,\ldots,b_{n-1}\}$ for the filler
items. In the figures, $a_0$ denotes the special agent placed at coordinate
$0$, and $a_\star$ denotes the other special agent; these are role labels, not
fixed identities.

\paragraph{Case of $k=1$:}
Choose $0<\varepsilon<1/10$ and let $L=10$. The predicted item
locations are
\[
    \widehat{P}(b_1)=1+\varepsilon,\qquad
    \widehat{P}(b_n)=3,\qquad
    \widehat{P}(b_j)=L+j\varepsilon \quad\text{for } j=2,\ldots,n-1.
\]
For each filler agent $a_j$ with $j\geq 3$, place $a_j$ at
\(\widehat{P}(b_{j-1})+\varepsilon/10\), with induced preferences. Both
special agents rank $b_1$ first, $b_n$ second, and then $B_F$; this list is
induced at either coordinate $0$ or coordinate $2$.

\begin{figure}[H]
    \centering
    \begin{tikzpicture}[
        x=1.08cm,y=0.9cm,
        agent/.style={circle,draw=cyan!60!black,fill=\agentcolor,inner sep=2pt},
        item/.style={circle,draw=black,fill=\itemcolor,inner sep=2pt},
        every node/.style={font=\scriptsize}
    ]
        \draw[black] (-0.35,0) -- (6.7,0);
        \node[agent,label=above:$a_0$] at (0,0) {};
        \node[item,label=below:$b_1$] at (1.2,0) {};
        \node[agent,label=above:$a_\star$] at (2.4,0) {};
        \node[item,label=below:$b_n$] at (3.6,0) {};
        \node at (4.6,0) {$\cdots$};
        \node[item,label=below:$B_F$] at (5.8,0) {};
        \node[agent,label=above:$A_F$] at (6.2,0) {};
    \end{tikzpicture}
    \caption{The $k=1$ item-location construction, with $A_F=\{a_3,\ldots,a_n\}$ and $B_F=\{b_2,\ldots,b_{n-1}\}$. The role labels $a_0$ and $a_\star$ are assigned after $F$ chooses the recipient of $b_n$.}
    \label{fig:item-location-k1}
\end{figure}
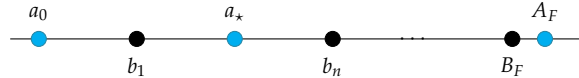

Run $F$ on this preference profile and on the item locations
$\widehat{P}$. If it assigns $b_n$ to $a_1$, place $a_1$ at $0$ and $a_2$ at $2$; if
it assigns $b_n$ to $a_2$, swap these two locations; and if it assigns $b_n$
to a filler agent, place $a_1$ at $2$ and $a_2$ at $0$. In all cases, the
same input remains consistent. Let $d$ be the resulting metric.

We first compute $\cost_1(\Mopt,d)$. Match $b_n$ to the special agent at
coordinate $2$, match $b_1$ to the special agent at coordinate $0$, and match
each filler item privately. The two special-agent costs are $1$ and
$1+\varepsilon$, and each filler cost is $\varepsilon/10$, so
$\cost_1(\Mopt,d)=1+\varepsilon$. Now let $M$ be the matching returned by
$F$ on this input. By construction, $M$ gives $b_n$ either to the special
agent at coordinate $0$, who pays $3$, or to a filler agent, who pays at least
$L-3>3$. Hence $\cost_1(M,d)\geq 3$, and therefore
$\distortion_1(M,d)\geq 3/(1+\varepsilon)$. Taking $\varepsilon\to0$ makes
this lower bound tend to $3$. This gives the desired lower bound for this
case.

\paragraph{Case of $k\geq 2$:}
Choose $0<\varepsilon<1/10$ and again let $L=10$. The predicted
item locations are
\[
    \widehat{P}(b_1)=\varepsilon,\qquad
    \widehat{P}(b_n)=2,\qquad
    \widehat{P}(b_j)=L+j\varepsilon \quad\text{for } j=2,\ldots,n-1.
\]
For each filler agent $a_j$ with $j\geq 3$, place $a_j$ at
\(\widehat{P}(b_{j-1})+\varepsilon/10\). The special agents both rank $b_1$
first, $b_n$ second, and then $B_F$; this list is induced at either
coordinate $0$ or coordinate $1$.

\begin{figure}[H]
    \centering
    \begin{tikzpicture}[
        x=1.08cm,y=0.9cm,
        agent/.style={circle,draw=cyan!60!black,fill=\agentcolor,inner sep=2pt},
        item/.style={circle,draw=black,fill=\itemcolor,inner sep=2pt},
        every node/.style={font=\scriptsize}
    ]
        \draw[black] (-0.35,0) -- (6.35,0);
        \node[agent,label=above:$a_0$] at (0,0) {};
        \node[item,label=below:$b_1$] at (0.65,0) {};
        \node[agent,label=above:$a_\star$] at (1.9,0) {};
        \node[item,label=below:$b_n$] at (3.0,0) {};
        \node at (4.0,0) {$\cdots$};
        \node[item,label=below:$B_F$] at (5.25,0) {};
        \node[agent,label=above:$A_F$] at (5.7,0) {};
    \end{tikzpicture}
    \caption{The $k\geq2$ item-location construction, with $A_F=\{a_3,\ldots,a_n\}$ and $B_F=\{b_2,\ldots,b_{n-1}\}$. The prediction does not reveal which special agent plays the role of $a_\star$.}
    \label{fig:item-location-kgeq2}
\end{figure}
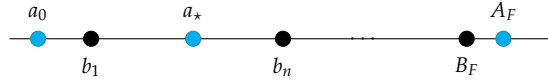

Run $F$ on this input. If it assigns $b_n$ to
$a_1$, place $a_1$ at $0$ and $a_2$ at $1$; if it assigns $b_n$ to $a_2$,
swap these two locations; and if it assigns $b_n$ to a filler agent, place
$a_1$ at $1$ and $a_2$ at $0$. Again, the same input remains consistent. Let
$d$ be the resulting metric.

We first compute an upper bound on $\cost_k(\Mopt,d)$. Match $b_n$ to the
special agent at coordinate $1$, match $b_1$ to the special agent at
coordinate $0$, and match each filler item privately. The two special-agent
costs are $1$ and $\varepsilon$, and each filler cost is $\varepsilon/10$, so
$\cost_k(\Mopt,d)\leq 1+(k-1)\varepsilon$. Now let $M$ be the matching
returned by $F$ on this input. If $M$ assigns $b_n$ to the special agent at
coordinate $0$, that agent pays $2$, and the other special agent receives an
item at distance at least $1-\varepsilon$. Since $k\geq2$,
$\cost_k(M,d)\geq 3-\varepsilon$. If $b_n$ goes to a filler agent, that agent
alone pays at least $L-2>3-\varepsilon$, so again
$\cost_k(M,d)\geq 3-\varepsilon$. Thus
\[
    \distortion_k(M,d)
    \geq
    \frac{3-\varepsilon}{1+(k-1)\varepsilon}.
\]
Taking $\varepsilon\to0$ makes this lower bound tend to $3$. This gives the
desired lower bound for this case.

Together, the two cases show that the consistency of $F$ is at least $3$ for
every $k\in[n]$.
\end{proof}

\end{document}